\documentclass[a4paper,UKenglish,cleveref, autoref, thm-restate]{lipics-v2021}

\usepackage{array}
\usepackage{thm-restate}
\usepackage{booktabs}
\usepackage[table,dvipsnames]{xcolor}

\nolinenumbers

\title{Reconfiguration of Squares Using a Constant Number of Moves Each}

\author{Thijs van der Horst}{TU Eindhoven and Utrecht University, the Netherlands}{t.w.j.v.d.horst@uu.nl}{}{}
\author{Maarten Löffler}{Utrecht University, the Netherlands}{m.loffler@uu.nl}{}{}
\author{Tim Ophelders}{TU Eindhoven and Utrecht University, the Netherlands}{t.a.e.ophelders@uu.nl}{}{}
\author{Tom Peters}{TU Eindhoven, the Netherlands}{t.peters1@tue.nl}{}{}

\authorrunning{T. van der Horst, M. Löffler, T. Ophelders and T. Peters} %TODO mandatory. First: Use abbreviated first/middle names. Second (only in severe cases): Use first author plus 'et al.'


\relatedversion{} %optional, e.g. full version hosted on arXiv, HAL, or other respository/website

\begin{document}

\maketitle
\begin{abstract}
    Multi-robot motion planning is a hard problem. We investigate restricted variants of the problem where square robots are allowed to slide over an arbitrary curve to a new position only a constant number of times each. We show that the problem remains NP-hard in most cases, except when the squares have unit size and when the problem is unlabeled, i.e., the location of each square in the target configuration is left unspecified.
\end{abstract}

\section{Introduction}
Multi-robot motion planning concerns problems where we are given a starting configuration of robots in the plane, as well as a desired target configuration. The goal is to move the robots from their starting location to the target configuration while avoiding collisions. Many different variants of these problems have been considered. Robots can have different sizes and shapes, they might have to avoid static obstacles in their workspace, and each robot can either have its own dedicated target, or robots can end up at arbitrary targets.

Except for some trivial cases, solving these problems optimally, minimizing either the makespan or path length of a solution, is hard.
In general, it is usually even hard to decide if a target configuration can be reached at all.
For a full survey of these results, see~\cite{dumitrescu2007motion}. Here, we give a brief overview of the known hardness results.

One of the first hardness results is from 1984, when it was shown that it is PSPACE-complete to decide if a set of arbitrary shaped objects can be reconfigured into a designated target configuration~\cite{hopcroft1984complexity}, this was then shown to already hold for rectangular objects moving in a rectangular region~\cite{hopcroft1986reducing}. The proof was later generalized, and used to show that rectangles of size $1\times 2$ and $2\times 1$ suffice~\cite{hearn2009games}. The well-known related \emph{Rush Hour Problem}, where rectangles are only allowed to slide into one direction, was also proven to be PSPACE-complete~\cite{flake2002rush}.
For unit squares in a simple polygon, deciding if a target configuration is reached is also PSPACE-hard, both if the polygon can have holes~\cite{brunner2021rushhour} as well as when it cannot~\cite{abrahamsen2024reconfiguration}.
For many moving disks it has been shown to be NP-hard to decide if a target configuration is reachable~\cite{spirakis1984strong}, even if the disks have the same size~\cite{dumitrescu2013reconfiguration}.

On the other hand, for two disks, a full characterization has been made of reconfiguration using minimum path length~\cite{kirkpatrick2016characterizing}, which has recently been extended to two arbitrary centrally symmetric objects~\cite{kirkpatrick2025minimum}. Moreover, if the objects are pebbles moving on a graph, which corresponds to squares moving in a grid aligned polygon, a polynomial time algorithm is known to calculate the reconfiguration using the least amount of moves if one exists~\cite{kornhauser1984coordinating}.

Most reconfiguration problems of this type are either hard, or a trivial algorithm exists. In an attempt to reduce the search space, \emph{monotone} movement schedules have been proposed: movement schedules where each robot is allowed to move only once. This open problem was recently stated by Abrahamsen and Halperin~\cite{abrahamsen2024ten}. If you furthermore only allow a single robot to move at each point in time, this reduces the problem to finding a valid permutation, hopefully reducing the complexity of the problem. In this paper, we investigate the multi-robot motion planning problem for square robots that are allowed to move at most a constant number of times. We show that in most cases, the decision problem remains NP-hard, except for some scenarios with unit size robots in which we can formulate the problem as a flow computation, and some trivial scenarios in which the robots are able to move arbitrarily far away and back again.

\section{Preliminairies}
\label{sec:definitions}
The \emph{multi-robot motion planning} problem takes a set of robots (or objects) in the Euclidean plane and tries to move them around to create a new configuration, while avoiding collisions. Let a \emph{configuration} be a set of non-overlapping squares in the Euclidean plane. Let a \emph{move} be a square sliding on an arbitrary curve in one continuous motion without rotating. A move is \emph{valid} if it does not collide (overlap) with any other square during the movement. Only a single square is allowed to move at any given time. A \emph{motion schedule} is a sequence of moves. A \emph{valid motion schedule} is a motion schedule where each move is valid. In the \emph{square reconfiguration} problem, we are given a starting configuration $S$, and a target configuration $T$ with $|S| = |T|$. The goal is to find any valid motion schedule that reconfigures $S$ into $T$, or decide that such a schedule does not exist.

We consider several variants of this problem, which we call \emph{scenarios}. In the \emph{labeled} scenarios, each square in $S$ has a dedicated target in $T$. In the \emph{unlabeled} scenarios, each square is allowed to end up at any target in $T$. In the \emph{bounded} scenarios the squares lie in the interior of a rectilinear polygonal domain with integer coordinates (potentially with holes). A move is only valid if, besides not intersecting other squares, it does not intersect the boundary of this domain at any time during a move. In the \emph{unbounded} scenarios, this polygonal domain does not exist and the squares move freely in plane (only colliding with other squares).

We also distinguish between the sizes of the squares. The squares can have a side length of $1$ unit, $2$ units, or more. In the literature, some papers consider the squares to always have unit side length, and instead allow the squares and the polygonal domain to use non-integer coordinates. These definitions are equivalent.
Lastly, we consider the number of moves that a square is allowed to perform. We only consider scenarios where each square is allowed to move at most $x$ times, where $x$ is a constant independent of the number of squares. If each square is allowed to perform only a single move, the problem and resulting movement schedule is also called \emph{monotonic}.

Note that in some scenarios the problem is trivial: in the \emph{unbounded} scenarios in which the number of allowed moves is at least $2$, we can simply move all robots arbitrarily far away one by one until our workspace is empty, and then move them back in any different order.

\paragraph*{Contributions}
In Section~\ref{sec:labeled-monotone}, we show that the labeled monotonic scenario is NP-hard. In Section~\ref{sec:1x1}, we give a polynomial time algorithm for the scenario in which the squares are unlabeled and have unit length. In Section~\ref{sec:unlabeled-monotonic}, we show how for squares that are bigger than unit length, the unlabeled monotonic scenario is also NP-hard. Lastly, in Section~\ref{sec:bounded-1x1-multiple-moves}, we show that the remaining scenario, of bounded squares of unit length allowed more than 1 move, is also NP-hard.
Hence, for each scenario, we show that the problem is either NP-hard, or, in the case of unlabeled unit-sized squares, we give a polynomial time algorithm to compute a solution. Our results can be found in Table~\ref{tab:results}.

\begin{table}[t]
\definecolor{mooirood}{HTML}{dd3333}
\definecolor{mooigroen}{HTML}{88eeaa}
\DeclareRobustCommand{\np}{\cellcolor{mooirood}\color{white}}
\DeclareRobustCommand{\p}{\cellcolor{mooigroen}\color{black}}
\begin{subtable}{0.5\linewidth}
\begin{tabular}{c|c|c|c|}
    & \multicolumn{3}{c|}{\# moves}
    \\
    size & 1 & 2 & $k$\\\hline
    $1\times 1$ & \np Thm.~\ref{thm:manta} & \np Thm.~\ref{thm:Hamiltonian-rendez} & \np Thm.~\ref{thm:Hamiltonian-rendez} \\\hline
    $2\times 2$ & \np Thm.~\ref{thm:manta} & \np Thm.~\ref{thm:rendez-2move} & \np Thm.~\ref{thm:rendez-latch} \\\hline
    $\ell\times \ell$ & \np Thm.~\ref{thm:manta} & \np Thm.~\ref{thm:rendez-2move} & \np Thm.~\ref{thm:rendez-latch} \\\hline
\end{tabular}
\subcaption{labeled, bounded}
\end{subtable}%
\qquad
\begin{subtable}{0.5\linewidth}
\begin{tabular}{c|c|c|c|}
    & \multicolumn{3}{c|}{\# moves}
    \\
    size & 1 & 2 & $k$\\\hline
    $1\times 1$ & \np Thm.~\ref{thm:manta} & \p trivial & \p trivial \\\hline
    $2\times 2$ & \np Thm.~\ref{thm:manta} & \p trivial & \p trivial \\\hline
    $\ell\times \ell$ & \np Thm.~\ref{thm:manta} & \p trivial & \p trivial \\\hline
\end{tabular}
\subcaption{labeled, unbounded}
\end{subtable}\\
\begin{subtable}{0.5\linewidth}
\begin{tabular}{c|c|c|c|}
    & \multicolumn{3}{c|}{\# moves}
    \\
    size & 1 & 2 & $k$\\\hline
    $1\times 1$ & \p Thm.~\ref{thm:flow} & \p Thm.~\ref{thm:flow} & \p Thm.~\ref{thm:flow} \\\hline
    $2\times 2$ & \np Thm.~\ref{thm:Hamiltonian} & \np Thm.~\ref{thm:rendez-2move} & \np Thm.~\ref{thm:rendez-latch} \\\hline
    $\ell\times \ell$ & \np Thm.~\ref{thm:Hamiltonian} & \np Thm.~\ref{thm:rendez-2move} & \np Thm.~\ref{thm:rendez-latch} \\\hline
\end{tabular}
\subcaption{unlabeled, bounded}
\end{subtable}%
\qquad
\begin{subtable}{0.5\linewidth}
\begin{tabular}{c|c|c|c|}
    & \multicolumn{3}{c|}{\# moves}
    \\
    size & 1 & 2 & $k$\\\hline
    $1\times 1$ & \p Thm.~\ref{thm:flow} & \p trivial & \p trivial \\\hline
    $2\times 2$ & \np Thm.~\ref{thm:Hamiltonian} & \p trivial & \p trivial \\\hline
    $\ell\times \ell$ & \np Thm.~\ref{thm:Hamiltonian} & \p trivial & \p trivial \\\hline
\end{tabular}
\subcaption{unlabeled, unbounded}
\end{subtable}
\caption{An overview of all scenarios and our corresponding results. On the horizontal axis is the number of moves per square (where $k$ is an arbitrary constant). On the vertical axis is the size of the squares (where $\ell$ is an arbitrary constant). \colorbox{mooigroen}{\p Green cells} allow a polynomial time solution. \colorbox{mooirood}{\np Red cells} are NP-hard.}
\label{tab:results}
\end{table}

\section{Labeled Monotonic Reconfiguration}\label{sec:labeled-monotone}

In the scenario with labeled squares, where each square is allowed only a single move (monotonic), square reconfiguration is NP-hard. We show this via a reduction from \textsc{Planar Monotone 3-SAT}, which is NP-hard~\cite{deberg2010optimal}. We will show the reduction for the bounded case. For the unbounded case, note that in any valid reconfiguration schedule, a square needs to move to its target in one move. Hence, if a square already is on its target location, it can never move. We can use this to replace the boundary. The reduction works for squares of arbitrary size.

For the variable gadget, we place $5$ squares representing true, and above that $5$ squares representing false in a horizontal hallway of height $2$, see Figure~\ref{fig:manta}. The labeling is such that the ``true'' squares should end up in the same order from left to right as they started, and similar for the ``false'' squares.
There are $3$ one-wide hallways above the false squares, and $3$ one-wide hallways below the true squares. Note that these hallways are only accessible if some of the corresponding false or true squares move to their target.

A clause gadget is simply a T-junction containing a target. All three arms of the T are connected to the corresponding variable gadgets on the ``true'' or ``false'' side.
The variable gadgets are connected horizontally with a start gadget on the left and a target gadget on the right.
The start gadget contains a single \emph{checker} square (yellow in Figure~\ref{fig:manta}). This square will only be able to move to its target once the variables have an assignment.
The end gadget contains $k$ starting \emph{clause} squares (purple), where $k$ is the number of clauses in the construction. The targets of these clause squares are the targets inside each clause gadget.
Behind these clause squares there is a single target to which the checker square needs to move.
Now the yellow checker square can only reach its target if and only if there is an assignment to the variables that satisfies the formula.

\begin{figure}[t]
    \centering
    \includegraphics{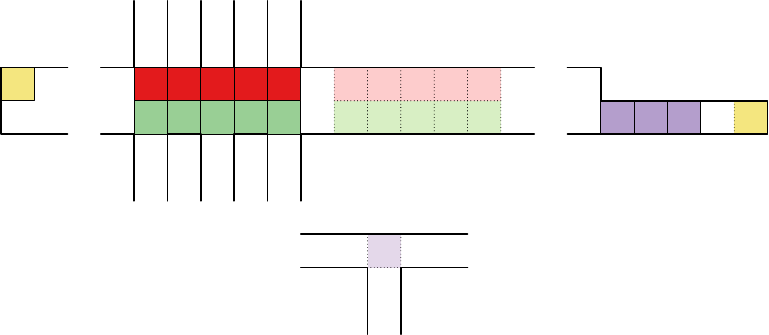}
    \caption{The gadgets for Theorem~\ref{thm:manta}. From left to right: The start gadget, the variable gadget, and the end gadget. Bottom: clause gadget. The green squares represent true, the red ones false. The dotted squares are targets.}
    \label{fig:manta}
\end{figure}

\begin{theorem}\label{thm:manta}
    The labeled, monotonic square reconfiguration problem is NP-hard.
\end{theorem}
\begin{proof}
    We will first show that the instance $M$ of the bounded, labeled, monotonic square reconfiguration problem with $1\times 1$ squares admits a valid movement schedule if and only if the corresponding \textsc{Planar Monotone 3-SAT} instance $I$ is satisfiable.

    If $I$ is satisfiable, we take for each variable its assignment. We move for each variable the corresponding true or false squares to their respective targets. Then, because the instance is satisfiable, each clause square has a path to its corresponding target. Now the central hallway is open and the checker square can move to its target. Lastly, the remaining variable squares move to their targets. This movement schedule is valid.

    Now assume that $I$ is not satisfiable. We will show that $M$ does not admit a valid movement schedule. For the checker square to move to its target, all the clause squares have to move first. Hence, the checker square can only ever move via the central hallway, since the clause gadgets will be blocked. Furthermore, for each variable, only the true, or only the false squares can move before the checker square, since otherwise it cannot pass. Hence, each variable needs an assignment. If $I$ is not satisfiable, there exists no assignment of the variables such that each clause square can reach its target. Hence, the checker square can never reach its own target and a valid movement schedule does not exist.

    To show that the construction works for unbounded instances, we notice that a boundary can be replaced by squares that are already on their target. These squares can never move. We also observe that the construction holds regardless of the size of the squares.
\end{proof}

\section{Unlabeled \texorpdfstring{$1\times 1$}{1x1}}\label{sec:1x1}
In the scenario with unlabeled $1\times 1$ squares, the key insight is that whenever a square $a$ wants to move to a target $a'$ via a path $\pi$, but it is blocked by another square $b$, we can take the last such $b$ and move it to $a'$ instead. We can use this to create an algorithm that finds a valid movement schedule if it exists. It creates a graph of all possible square positions and creates a flow-graph out of it, having flow from the starting locations to the target locations. This flow is then used to create a movement schedule.

Because our squares are $1\times 1$, we consider only integer coordinates. Consider the set $C$ of all integer coordinates that are inside the polygonal domain (if bounded) or inside the axis aligned bounding box of $S$ and $T$ (if unbounded). Let $G$ be the grid graph where the vertices correspond to $C$ and there is an edge between two vertices $u$ and $v$ if their corresponding integer coordinates are horizontally or vertically adjacent.
If $G$ contains more than one connected component, these components can be solved independently from each other and the full problem admits a valid movement schedule if and only if each individual component admits a valid movement schedule.
Hence, from this point, we assume that $G$ consists of a single connected component.
We will show that in this case, there is always a solution.

To do so, we make a flow graph $G_F$ out of $G$ as follows. We add a single source and a single target. The source is connected with a directed edge with capacity $1$ to each vertex that represents a coordinate in $S$ and the target has a directed edge with capacity $1$ from each vertex representing a coordinate in $T$. All edges in $G$ are replaced by directed edges in $G_F$ with a capacity of $n$.

We compute a maximal flow~$F$ on $G_F$~\cite{goldberg1989network}. Because the singular source and singular target only have edges with capacity of $1$, this maximal flow contains only integer values for each edge. The maximal flow $F$ might contain cycles. If $F$ contains a positive weight directed cycle, we reduce the flow value on each edge in the cycle by $1$. We repeat this until $F$ contains no cycles. When removing the singular source and singular sink from $F$, what remains is a Directed Acyclic Graph (DAG) $F'$ with positive edge weights. This DAG $F'$ represents the flow of squares from $S$ to $T$ over the grid cells.

To construct a movement schedule out of $F'$, we take any arbitrary sink $t$ in $F'$, i.e. a node with out-degree zero. Sink $t$ represents a target square in $T$. We now find a path in $F'$ from any vertex $s$ representing a start in $S$ to $t$, such that this path does not visit any other vertex in $S$. We reduce the weight on all edges on this path by one and move $s$ to $t$ via the path. We repeat this until no such sinks $t$ remain.

\begin{theorem}\label{thm:flow}
    The unlabeled, square reconfiguration problem for $1\times 1$ squares has a polynomial time algorithm.
\end{theorem}
\begin{proof}
    To show that our algorithm provides a correct solution, we need to show that there always exists at least one sink $t$ in $F'$, and that there always exists a corresponding path from a source $s$ to $t$ that does not visit any other vertex in $S$.

    First we show that $F'$ stays a DAG with a single connected component throughout the algorithm. $F'$ is a DAG at the start of the algorithm, since it is constructed from a flow network that had its cycles removed. At this point, $F'$ is a valid flow network, i.e. for every vertex that is not in $S$ or $T$, the inflow is equal to the outflow. Whenever we remove flow from the edges that make up a path from $s$ to $t$, we remove one inflow and one outflow from every vertex on its path (only one outflow from $s$ and one inflow from $t$). Hence, it remains a valid flow network.

    Because $F'$ contains no cycles, there must be at least one sink $t\in T$ without outflow.
    Because in $F$ all vertices in $T$ are connected to the sink with an edge with weight $1$, the inflow of $t$ is $1$.
    Because $F'$ is a valid flow network, there must be a path from some source $s\in S$ to $t$.
    Assume for contradiction that this path visits some other source $s'\in S$.
    Let $s'$ be the last source that this path visits.
    Now, we can pick the path from $s'$ to $t$ instead.
\end{proof}

The algorithm works regardless of if there is a boundary or not, and uses only a single move per square. Hence, if there exists any valid movement schedule, there also exists one that uses only a single move per square.

\section{Unlabeled, Big Squares, Monotonic}\label{sec:unlabeled-monotonic}
In the scenario with unlabeled squares that are bigger than $1\times 1$, we will show that the square reconfiguration problem is NP-hard if each square is allowed to move only once, by a reduction from the \textsc{Hamiltonian Path} problem. This problem is NP-hard even for planar directed graphs that have maximum degree $3$, and when we designate the start and target vertex of the path that each have degree $1$~\cite{itai1982hamilton}. Note that each vertex needs at least one incoming and at least one outgoing edge, otherwise the answer is trivial. A vertex with two incoming and one outgoing edge is called a merge vertex, and a vertex with two outgoing and one incoming edge is called a split vertex.

Given an instance of \textsc{Hamiltonian Path}, we will construct a corresponding instance of the unlabeled square reconfiguration problem. We replace all (split and merge) vertices in the graph by (split and merge) vertex gadgets. We replace the edges by edge gadgets, see Figure~\ref{fig:Hamiltonian}. Furthermore, we have a start and an end gadget, representing the start and end of the Hamiltonian path.

The gadgets are surrounded by \emph{blocker} squares (yellow in Figure~\ref{fig:Hamiltonian}). These squares form a boundary and can never move. The green squares represent edges, and the two purple squares represent the vertex. Because of the way the purple vertex squares are positioned, they cannot reconfigure themselves. Hence, the only way to reconfigure the vertex gadgets, is if one of the green squares moves away and makes space. An example of such a reduction can be seen in Figure~\ref{fig:Hamiltonian-instance}.
% For the merge gadget, this needs to be a green square representing a horizontal edge, for the split gadget, this needs to be a green square representing the vertical edge. In both cases these squares represent incoming edges to the vertex. Then, one purple square can move out of the way, making space for the other purple square to move to one of the targets it covers.

\begin{figure}[t]
    \centering
    \includegraphics{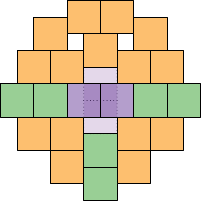}\hfil
    \includegraphics{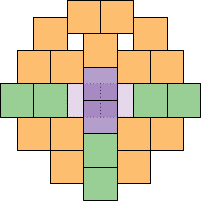}\hfil
    \includegraphics{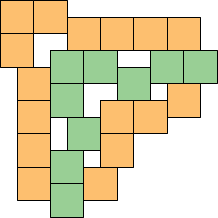}\\[1em]
    \includegraphics{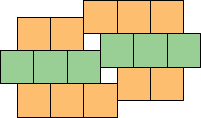}\hfil
    \includegraphics{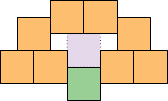}\hfil
    \includegraphics{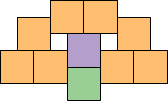}
    \caption{The gadgets for Theorem~\ref{thm:Hamiltonian}. In reading order: the merge, split, corner, wire (with a shift), start, and end gadget. For the merge gadget, the two horizontal edges are incoming. For the split gadget, the two horizontal edges are outgoing. Filled squares are starting locations. The dashed squares are target locations. Both the green as well as the yellow squares cover a start as well as a target square.}
    \label{fig:Hamiltonian}
\end{figure}

\begin{figure}[p]
    \centering
    \includegraphics{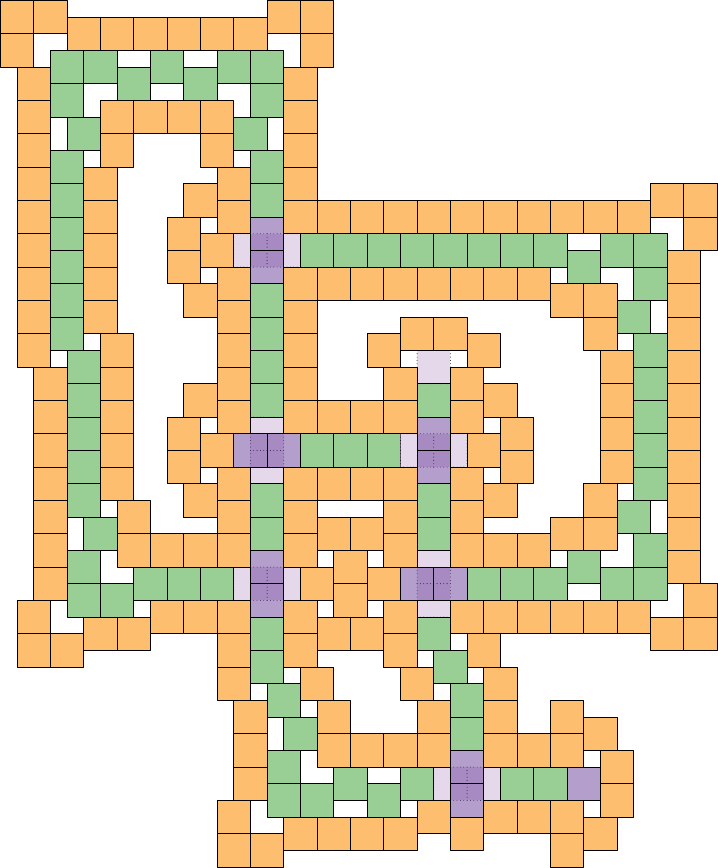}\\[1em]
    \includegraphics{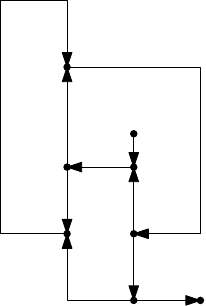}
    \caption{An instance of our Hamiltonian path reduction with a valid movement schedule. Top: The square reconfiguration problem, bottom: The corresponding directed graph.}
    \label{fig:Hamiltonian-instance}
\end{figure}

\begin{theorem}\label{thm:Hamiltonian}
    The unlabeled, bounded monotonic square reconfiguration problem with squares of size $2$ or bigger is NP-hard.
\end{theorem}
\begin{proof}
    We will show that the instance $M$ of the unlabeled, monotonic square reconfiguration problem with squares of side length at least $2$ admits a valid movement schedule if and only if the corresponding instance $I$ of the \textsc{Hamiltonian path} problem has a Hamiltonian path.

    We will show that at no point, a yellow blocker square can move. No blocker square partially overlaps with a target. Hence, for any blocker to move, another target needs to be completely uncovered. Furthermore, there should be a valid path from a blocker square to this target.

    We will then also show that any valid movement schedule corresponds to a valid Hamiltonian path~$\pi$ in the original graph $G$.
    
    In the initial configuration, only the start gadget has a target position that is not obstructed by a square. Only the green square can reach this target.
    From that point on, in the edge gadgets, each target uncovered by a green square can only be reached by the next green square, until we reach a vertex gadget.
    This edge is the first edge of $\pi$.

    The vertex gadget can be either a split or a merge gadget.
    Regardless of if it is a split gadget or a merge gadget, only one of the purple squares can move into the target last uncovered by a green edge square.

    Then, the other purple square can move into one of the targets of the vertex gadget. For the merge gadget, this needs to be the top target. For the split vertex, there is a choice. If we pick the left target, the right target can only be filled by the green edge square to the right and vice versa. From this point on, we can only move a green edge square into the newly uncovered target.

    This traces a path through our vertex gadgets, which corresponds to a path in the original graph.

    Because the vertex gadgets cannot reconfigure themselves, and because the purple square in the end gadget needs to move, this path visits every vertex gadget. Moreover, because every square can move at most once, this path visits every vertex at most once.
    Hence, this path is a Hamiltonian path.
\end{proof}

Note that this construction works for all unbounded scenarios. Hence, the bounded version is also NP-hard, since we can add a large bounding box without affecting the result.

\section{Bounded, Big Squares, Multiple Moves}
In this section we will show that every scenario of the square reconfiguration problem is NP-hard in which we consider squares of at least $2\times 2$ or bigger, that can make a constant amount of moves $k$ with $k\geq 2$, and when the squares live inside a polygonal domain.

We will first assume that the squares can move only twice. To show that the bounded version of the square reconfiguration problem is NP-hard, we will perform another reduction from \textsc{Planar Monotone 3-SAT}. We replace the variables, clauses, and edges by gadgets as shown in Figure~\ref{fig:rendez}. Every square starts on a target (its own in the labeled case), except for one purple \emph{clause checker square} per clause gadget. The squares need to make space for all clause checker squares to move, and then move back to their starting location.

\begin{figure}[b]
    \centering
    \includegraphics{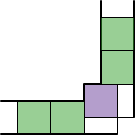}\hfil
    \includegraphics{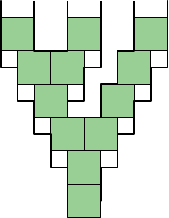}\hfil
    \includegraphics{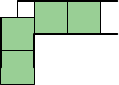}\hfil
    \includegraphics{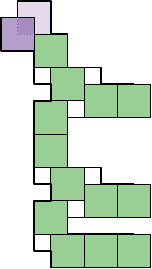}
    \caption{The gadgets for Theorem~\ref{thm:rendez-2move}. From left to right: Variable gadget, splitter gadget, wire gadget with a corner, and clause gadget. All squares have a target below them, except for the clause checker square.}
    \label{fig:rendez}
\end{figure}

\begin{theorem}\label{thm:rendez-2move}
    The bounded, square reconfiguration problem with squares of side length at least $2$ is NP-hard if each square is allowed at most $2$ moves. 
\end{theorem}
\begin{proof}
    We will show that the instance $M$ of the square reconfiguration problem with squares of side length at least $2$ has a valid movement schedule if and only if the corresponding instance $I$ of \textsc{Planar Monotone 3-SAT} has a variable assignment that satisfies the formula.

    We first observe that each square is already at the target location, except for the clause checkers. For the clause checker squares to move to (their) target location, the squares in at least one hallway must make place. For this to happen, a corresponding variable gadget needs to move to the corresponding false or true position.

    Each square is allowed to move only twice. Hence, each variable gadget can only go to either the true or to the false position before it needs to move to a target. The only way for a clause checker to move to a target is if a corresponding variable gadget moves to the correct assignment. Since each variable can only perform one of these assignments, if there is no assignment that satisfies the clause, the clause checker will not move. Therefore, if there is no satisfiable assignment, there is also no valid movement schedule.
\end{proof}

Once a square is allowed to move more than twice, the current variable gadget can get multiple assignments.
We will change the variable gadget when a square is allowed to make $k$ moves, with $k > 2$.

To do so, we use a \emph{latch} gadget, see Figure~\ref{fig:rendez-latch}. This gadget can ``waste'' moves. To be more precise: if we want to move $k$ squares through the latch, the latch needs to move at least $k$ times (or $k-1$ if $k$ is even, since two squares can move through the latch at the same time) and if the latch is allowed to move $k$ times, at most $k$ squares can move through the latch (or $k+1$ if $k$ is odd).

% \begin{figure}[t]
%     \centering
%     \includegraphics{rendez-latch}
%     \caption{The latch gadget. Any square moving horizontally through the hallway cannot do so in a singular move.}
%     \label{fig:latch}
% \end{figure}

Using this latch gadget we create a variable gadget that can only have a single assignment, see Figure~\ref{fig:rendez-latch}. To move either the true or false squares inwards, sufficiently many of the squares below need to move through the latch. The latch can only move $k$ times. We design the variable in such a way that exactly the right amount of squares need to move through the latch for either the true or the false square to be able to move. After they have gone back, the latch has moved $k$ (or $k - 2$ depending on the starting location of the latch) times. The latch cannot move enough times to allow the variable to have a different assignment.

\begin{figure}[tp]
    \centering
    \includegraphics[width=\linewidth]{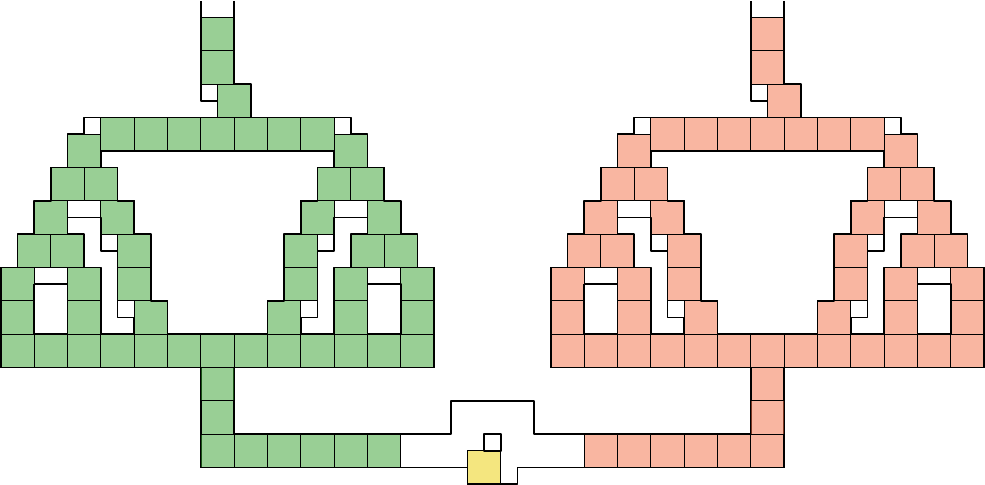}\\[1em]
    \includegraphics[width=\linewidth]{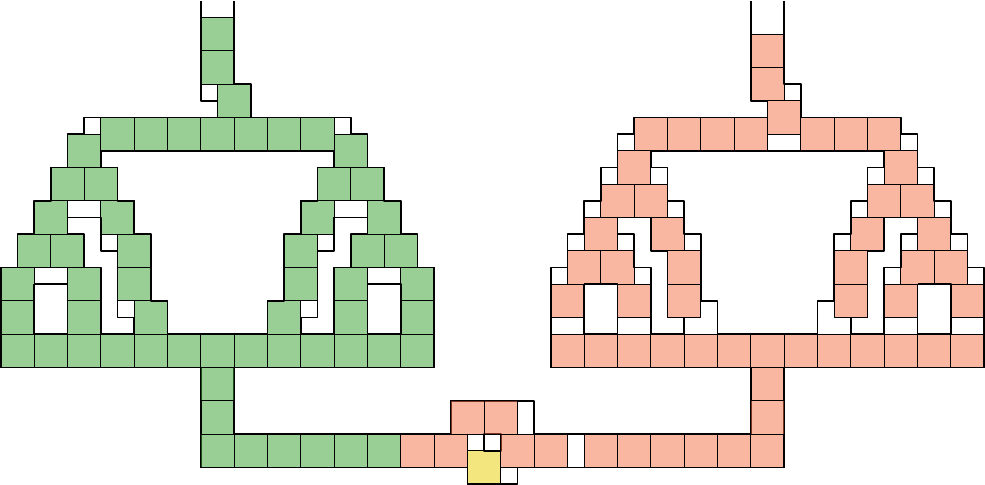}
    \caption{The variable gadget where the yellow latch square can move at most $6$ times. Setting the variable to false (bottom) forces the latch square to move twice. Resetting it takes another two moves. Similarly, setting the variable to true takes three moves, and resetting it takes another three.}
    \label{fig:rendez-latch}
\end{figure}

\begin{theorem}\label{thm:rendez-latch}
    The bounded, square reconfiguration problem with squares of side length at least $2$ is NP-hard if each square is allowed at most $k$ moves, for any constant $k$.
\end{theorem}

\section{Bounded, Labeled \texorpdfstring{$1\times 1$}{1x1}, Multiple Moves}\label{sec:bounded-1x1-multiple-moves}
The last remaining scenario is when we have squares of size $1\times 1$, each with a predetermined target location, moving inside a polygonal domain. We will again use a reduction from the \text{Hamiltonian Path} problem to show NP-hardness. As in Theorem~\ref{thm:Hamiltonian}, an empty square $e$ will trace a path in the configuration. As in Theorem~\ref{thm:rendez-2move}, the squares need to end on the same location as they started. This way, the reduction works for labeled reconfiguration.

We will again show the hardness of this problem via a reduction from the \textsc{Hamiltonian Path} problem on a directed graph with degree at most $3$. Again, as in Theorem~\ref{thm:Hamiltonian}, almost all locations are initially filled with squares, except for the starting vertex, which contains an empty square $e$. Square $e$ will trace a path in the construction that corresponds to a Hamiltonian path in the original graph. However, as in Theorem~\ref{thm:rendez-2move}, most squares start on their own target location. Hence, after $e$ has traced a path through the graph, it needs to go the exact same way back to ensure each square ends where it should.

Given an instance $I$ of the \textsc{Hamiltonian Path} problem, we will construct an instance $M$ of the right square reconfiguration problem as follows. We replace all vertices in the graph by vertex gadgets as shown in Figure~\ref{fig:Hamiltonian-rendez-vertex}, and each directed edge by an edge gadget in the right orientation as shown in Figure~\ref{fig:Hamiltonian-rendez-edge}. Furthermore, we have a start and an end gadget, which is the same as in Theorem~\ref{thm:Hamiltonian}, except with the blocker squares replaced by the bounding polygon.

\begin{figure}
    \centering
    \includegraphics{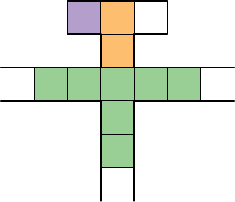}
    \caption{The vertex gadget for the Hamiltonian path construction for labeled squares of size $1$ inside a polygonal domain. Every square starts on its target, except the purple one, which needs to cross from left to right.}
    \label{fig:Hamiltonian-rendez-vertex}
\end{figure}

\begin{figure}[t]
    \centering
    \includegraphics{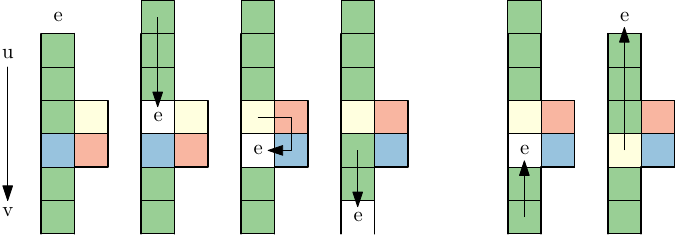}
    \caption{The directed edge gadget representing a directed edge from $u$ to $v$ with the (only) possible reconfiguration sequence. On the left the start configuration, on the right the target configuration. The first 3 steps are when the empty square $e$ traverses the edge downwards, the last two are when the empty square traverses the edge upwards. In between, $e$ can visit $v$, but does not have to.}
    \label{fig:Hamiltonian-rendez-edge}
\end{figure}

\begin{theorem}
    The bounded, labeled, square reconfiguration problem with squares of side length 1 is NP-hard if each square is allowed at most 2 moves.
\end{theorem}
\begin{proof}
    We will show that the instance $M$ of the square reconfiguration problem with squares of side length 1 has a valid movement schedule if and only if the corresponding instance $I$ of the \textsc{Hamiltonian Path} problem has a Hamiltonian path. We first observe that the vertex gadgets, nor the edge gadgets can reconfigure by themselves. Let $e$ be the empty space in the start gadget. In any valid reconfiguration schedule, $e$ needs to visit each vertex gadget, and each edge gadget at least once.

    The edge gadgets are directed from vertex $u$ to vertex $v$. When $e$ visits an edge, we call that visit \emph{nontrivial} if it leaves the configuration different than before the visit.
    Each edge can only be visited nontrivially twice, because each square is allowed to move only twice.
    Each edge needs at least one nontrivial visit.
    When $e$ visits an edge, if it exits the edge in a different direction than it came in, that visit is nontrivial.
    For a directed edge $(u, v)$, if the first visit is from $v$ to $u$, then there is no way to reach the target configuration.
    The edge can only be reconfigured if $e$ visits the edge from $u$ to $v$, and later back from $v$ to $u$, or if $e$ visits the edge from $u$, reconfigures the edge, and leaves via $u$ again.

    If a Hamiltonian path $\pi$ exists, we can construct a valid movement schedule by moving the empty square over the edges of $\pi$. Each time it visits a vertex, it is reconfigured, each time it visits an edge, we perform the reconfiguration as in Figure~\ref{fig:Hamiltonian-rendez-edge}. Once the empty square reaches the target vertex, it moves back on $\pi$. Now, every time it visits a vertex, it also visits all edges neighboring this vertex that were not part of $\pi$, but just far enough to reconfigure the edge gadget. Once the empty square reaches the start again, all squares are at their target position and have moved at most twice. This is a valid movement schedule.

    Now assume that a valid movement schedule exist. We will show that it corresponds to a valid Hamiltonian path $\pi$. The created path $\pi$ will be visiting the same vertices as $e$ in the same order. Moreover, $e$ visits each vertex at least once. We only consider edge visits where $e$ visits the edge for the first time, and exits the edge via a different direction than it entered. As shown before, since the movement schedule is valid, this can only happen when $e$ enters a directed edge $(u, v)$ from $u$. Hence, these visits of edges trace out a valid path in the corresponding graph. Each edge that is used by $e$ to move from one vertex to another must eventually be visited by $e$ in the opposite direction, and each edge can only be visited nontrivially at most twice. Combined with the fact that the degree of each vertex is three, the path traced by concatenating the first such edge visits cannot contain a cycle, but must visit each vertex at least once. Hence, this is a valid Hamiltonian path.
\end{proof}

To generalize this to problem instances where squares are allowed to move at most $k = O(1)$ times, for $k > 2$, we need to adjust the edge gadget accordingly. We extend the loop that is attached to the edge to a length at least $k$, see Figure~\ref{fig:hamiltonian-rendez-edge-4}. The target configuration has all squares in the loop shifted by $k$ positions. To prevent the empty spot of the vertex gadgets from interfering, we make sure that the distance between the vertex gadgets and the loop in each edge gadget is at least $k$. Now $e$ must go around the loop in each edge at least $k$ times for the reconfiguration to work. Again, if the first nontrivial visit is from $v$ to $u$, reconfiguration is impossible. Hence, the proof above still works.

\begin{theorem}\label{thm:Hamiltonian-rendez}
    The bounded, labeled, square reconfiguration problem with squares of side length $1$ is NP-hard if each square is allowed at most $k$ moves, for any constant $k\geq 2$.
\end{theorem}

\begin{figure}[b]
    \centering
    \includegraphics{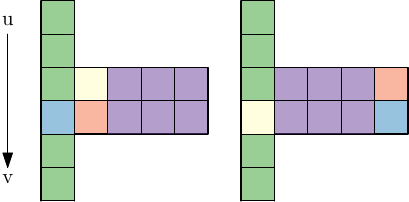}
    \caption{The edge gadget for the reconfiguration problem when each square is allowed to move at most $4$ times. Left: the starting configuration. Right: The target configuration.}
    \label{fig:hamiltonian-rendez-edge-4}
\end{figure}

\newpage

\section{Conclusion}
We have shown that deciding if reconfiguration of squares is possible is NP-hard in most scenarios when squares are only allowed a constanct number of moves. The only exception are unlabeled, unit size squares. The algorithm that we gave is not necessarily optimal; we only considered the decision variant of the problem. An interesting question is if optimal movement schedules for these problems can be calculated. Another open question is if these results can be extended to consider other object types, such as disks. Lastly, our reductions use polygonal domains that have holes. This is a natural result in reductions from graph problems. It is currently unknown if these results still hold in polygonal domains without holes.

\bibliography{thebibliography}

\appendix
\end{document}